\documentclass[]{eptcs}
\providecommand{\event}{TARK 2017} 
\usepackage{breakurl} 
\usepackage{underscore}           
\usepackage{amssymb,amsmath,amsthm,tikz}
\usepackage[english]{babel}
\usetikzlibrary{calc}
\usepgflibrary{arrows}
\usepackage{caption,subcaption}
\usepackage{verbatim,mathdots}
\usepackage{rotating}
\usepackage{array}
\newtheorem{corollary}{Corollary}
\newtheorem{lemma}{Lemma}
\newtheorem{theorem}{Theorem}

\theoremstyle{definition}
\newtheorem{definition}{Definition}
\theoremstyle{remark}
\newtheorem{remark}{Remark}

\newcommand{\agents}{\mathcal{A}}

\newcommand{\aaul}{\updownarrow}
\newcommand{\refa}{\mathit{ref_a}}

\newcommand{\M}{\mathcal{M}}
\renewcommand{\phi}{\varphi}

\title{Arbitrary Arrow Update Logic with Common Knowledge is neither RE nor co-RE}
\author{Louwe B. Kuijer
\institute{Department of Computer Science, University of Liverpool, UK}
\institute{LORIA, CNRS/Universit\'e de Lorraine, France}
\email{louwe.kuijer@liverpool.ac.uk}}

\def\titlerunning{AAULC is neither RE nor co-RE}
\def\authorrunning{L. B. Kuijer}

\begin{document}

\maketitle
\begin{abstract}
Arbitrary Arrow Update Logic with Common Knowledge (AAULC) is a dynamic epistemic logic with (i) an arrow update operator, which represents a particular type of information change and (ii) an arbitrary arrow update operator, which quantifies over arrow updates.

By encoding the execution of a Turing machine in AAULC, we show that neither the valid formulas nor the satisfiable formulas of AAULC are recursively enumerable. In particular, it follows that AAULC does not have a recursive axiomatization.
\end{abstract}

\section{Introduction}
One of the active areas of study in the field of Dynamic Epistemic Logic is that of \emph{quantified update logics}. Examples of these quantified logics include Arbitrary Public Announcement Logic (APAL) \cite{APAL}, Group Announcement Logic (GAL) \cite{Agotnes2010}, Coalition Announcement Logic (CAL) \cite{CAL}, Arbitrary Arrow Update Logic (AAUL) \cite{AAULAIJ}, Refinement Modal Logic (RML) \cite{refinement} and Arbitray Action Model Logic (AAML) \cite{hales13}.

All these logics have an operator that quantifies over all updates of a particular type. For example, in APAL the formula $[!]\phi$ means ``$[\psi]\phi$ holds for every public announcement $\psi$'' and in AAUL the formula $[\aaul]\phi$ means ``$[U]\phi$ holds for every arrow update $U$.''

One important question about these logics is the decidability of their satisfiability problems. The satisfiability problems of RML and AAML are known to be decidable \cite{hales13,refinement}, whereas for the other logics the satisfiability problem is known to be undecidable \cite{french08,agotnes2014,agotnesetal:2016,undecidable}. More precisely, the satisfiability problems for each of these undecidable logics was shown, by a reduction from the tiling problem, to be co-RE hard. But, so far, it has remained an open question whether they are co-RE.

In other words, while we know that we cannot generate a list of all satisfiable formulas of APAL, GAL, CAL and AAUL, we do not know whether it is possible to generate a list of all valid formulas of these logics. 

The question of whether the valid formulas are RE is of particular interest, since a negative result would imply the non-existence of a recursive axiomatization of the logic in question.\footnote{Finitary axiomatizations for APAL and GAL were proposed in \cite{APAL} and \cite{Agotnes2010}, respectively, but these axiomatizations contain a flaw that renders them unsound.\footnotemark See \url{http://personal.us.es/hvd/errors.html} for details and a proof of the unsoundness.}\footnotetext{We should stress that it is only the finitary axiomatizations that are unsound; the infinitary axiomatization presented in \cite{APAL} is sound and complete (although the completeness proof contains an error; again, see \url{http://personal.us.es/hvd/errors.html} for details).} After all, a recursive axiomatization would allow us to list all valid formulas.

Here, we study a variant AAULC of AAUL, which in addition to all the operators from AAUL also contains a common knowledge operator. We show that the valid formulas of AAULC are not recursively enumerable, by a reduction from the non-halting problem. The proof from \cite{undecidable}, which shows that the validities of AAUL are not co-RE also applies to AAULC. Still, the non-RE proof in this paper can be extended to a non-co-RE proof with little effort, so in this paper we prove that the validity problem of AAULC is neither RE nor co-RE.

We consider this result to be interesting in its own right. Additionally, and perhaps even more importantly, we also hope that the proof presented here can provide inspiration for proofs about the (non)existence of recursive axiomatizations for APAL, GAL, CAL and AAUL.

The structure of this paper is as follows. First, in Section~\ref{sec:aaul}, we define AAULC. Then, in Section~\ref{sec:machines} we discuss the notation that we will use to describe Turing machines. Finally, in Section~\ref{sec:reduction}, we show that both the halting problem and the non-halting problem can be reduced to the validity problem of AAULC.

\section{AAULC}
\label{sec:aaul}
Here, we provide the definitions of Arbitrary Arrow Update Logic with Common Knowledge (AAULC). The logics AUL, AAUL and AAULC were designed to reason about information change, but they can also be applied to other domains, most notably that of Normative Systems. A brief overview of the epistemic interpretation of arrow updates is given after the formal definitions. See \cite{AUL} and \cite{AAULAIJ} for a more in-depth discussion of the applications of AUL and its variants.

Let $\mathcal{P}$ be a countable set of propositional atoms, and let $\agents$ be a finite set of agents. We use five agents in our proof, so we assume that $|\agents|\geq 5$. The proof can be modified to use only one agent, but such modification requires a lot of complicated notation so we do not so here.
\begin{definition}
The language $\mathcal{L}_\mathit{AAULC}$ of AAULC is given by the following normal forms.
\begin{align*}
\phi ::= {} & p \mid \neg \phi \mid \phi \vee \phi \mid \square_a\phi \mid C\phi \mid [U]\phi \mid [\aaul]\phi\\
u := {} & (\phi,a,\phi)\\
U := {} & \{u_1,\cdots,u_n\}
\end{align*}
Where $p\in \mathcal{P}$ and $a\in \agents$. The language $\mathcal{L}_\mathit{AULC}$ of \emph{Arrow Update Logic with Common Knowledge} is the fragment of $\mathcal{L}_\mathit{AAULC}$ that does not contain the $[\aaul]$ operator. 
\end{definition}
We use $\wedge,\lozenge$ and $\langle \aaul\rangle$ in the usual way as abbreviations. 
The formulas of $\mathcal{L}_\mathit{AAULC}$ are evaluated on standard multi-agent Kripke models.
\begin{definition}
A \emph{model} is a triple $\M=(W,R,V)$, where $W$ is a set of worlds, $R:\agents\rightarrow 2^{W\times W}$ assigns to each agent an accessibility relation and $V:\mathcal{P}\rightarrow 2^W$ is a valuation.
\end{definition}
Note that we use the class $K$ of all Kripke models. Our reason for using $K$, as opposed to a smaller class such as $S5$, is that Arrow Update Logic is traditionally evaluated on $K$, see also \cite{AUL} and \cite{AAULAIJ}. For the results presented in this paper the choice of models is not very important; the proof that we use would, with some small modifications, also work on $S5$.

We also write $R_a(w)$ for $\{w'\mid (w,w')\in R(a)\}$. The semantics for most operators are as usual, so we omit their definitions. We do provide definitions for $[U]$ and $[\aaul]$, since these operators are not as well known as the others.
\begin{definition}
Let $\M=(W,R,V)$ be a model, and let $w\in W$. Then
$$\begin{array}{llr}
	\M,w\models [U]\phi & \Leftrightarrow& \M*U,w\models \phi\\
	\M,w\models [\aaul]\phi & \Leftrightarrow& \forall U\in \mathcal{L}_\mathit{AULC}: \M,w\models [U]\phi 
\end{array}$$
where $\M*U$ is given by
\begin{align*}
	\M* U := {} & (W,R*U,V),\\
	R*U(a) := {} & \{(w_1,w_2)\in R(a)\mid \\ & \hspace{10pt}\exists (\phi_1,a,\phi_2)\in U: \M,w_1\models \phi_1 \text{ and } \M,w_2\models \phi_2\}
\end{align*}
\end{definition}

\begin{definition}
Let $\M_1=(W_1,R_1,V_1)$ and $\M_2= (W_2,R_2,V_2)$ be models and let $w_1\in W_1$, $w_2\in W_2$. We say that $w_1$ and $w_2$ are $\mathit{AULC}$-indistinguishable if for every $\phi\in \mathcal{L}_\mathit{AULC}$, we have $\M,w_1\models \phi \Leftrightarrow \M,w_2\models \phi$.
\end{definition}
An arrow update $U$ represents an information-changing event, of a kind that is sometimes referred to as a \emph{semi-private announcement}. Unlike with a public announcement, the information gained though a semi-private announcement is not common knowledge. It is, however, common knowledge what information is gained under which conditions.

A typical example is the following. Suppose that $a$ and $b$ are playing a game of cards, where each player holds one card. The cards have been dealt to them, face down. Now, $a$ picks up her card and looks at it. By doing this, agent $a$ learns what card she holds. This new information is not common knowledge, since $b$ doesn't learn $a$'s card, so this event cannot be represented as a public announcement. It is common knowledge, however, under which conditions $a$ gains which information: if $a$ has the Ace of Spades, then she learns that she has the Ace of Spades, and so on. The event of $a$ picking up her card can therefore be considered a semi-private announcement, which can be represented as an arrow update.

A clause $(\phi_1,a,\phi_2)\in U$ says that in every world that satisfies $\phi_1$, the new information gained by agent $a$ is consistent with $\phi_2$. If there are multiple clauses that apply to a single world, we consider them to apply disjunctively, i.e., the new information is consistent with both postconditions: if $(\phi_1,a,\phi_2), (\psi_1,a,\psi_2)\in U$ and a world satisfies both $\phi_1$ and $\psi_1$, then $a$'s new information is consistent with every world that satisfies either $\phi_2$ or $\psi_2$. We assume that $U$ provides a full description of the new information, so any world consistent with the new information satisfies the postcondition of at least one applicable clause.

Semantically, this means that a transition $(w_1,w_2)\in R(a)$ is retained by the update $[U]$ if and only there is at least one clause $(\phi_1,a,\phi_2)\in U$ such that $w_1$ satisfies $\phi_1$ and $w_2$ satisfies $\phi_2$. Every other transition is removed from the model.

The example discussed above, where $a$ looks at her card, is represented by the arrow update 
\begin{equation*}U_\mathit{cards} := \{(\top,b,\top)\}\cup \{(\mathit{card},a,\mathit{card})\mid \mathit{card}\in\mathit{deck}\},\end{equation*}
where $\mathit{deck}$ is the deck from which the cards were dealt. The clause $(\top,b,\top)$ states that $b$ doesn't directly learn anything new: every distribution of cards is consistent with $b$'s new information. The clause $(\mathit{card},a,\mathit{card})$, for $\mathit{card}\in \mathit{deck}$ states that if $a$ holds $\mathit{card}$, then by looking at her card she learns that she holds $\mathit{card}$.

The arbitrary arrow update operator $[\aaul]$ quantifies over all arrow updates that do not themselves contain the $[\aaul]$ operator. So $\M,w\models [\aaul]\phi$ if and only if $\M,w\models [U]\phi$ for every $U\in \mathcal{L}_{AULC}$. This restriction to $[\aaul]$-free updates keeps the semantics from becoming circular.\footnote{Similar restrictions exist in the other quantified update logics. In APAL, for example, the arbitrary arrow update operator $[!]$ quantifies over all public announcements $[\psi]$ where $\psi\in \mathcal{L}_\mathit{PAL}$.}

The operator $[\aaul]$ allows us to ask, inside the object language, whether there is a semi-private announcement that makes a formula true. So, for example, $\M,w\models [\aaul](\square_ap \wedge \neg \square_bp)$ asks whether, in the situation represented by the model $\M,s$, there is a semi-private announcement that informs $a$ of the truth of $p$ without letting $b$ know that $p$ is true. 
Recall that an event is a semi-private announcement if it is common knowledge under what conditions which information is gained. So $[\aaul](\square_ap\wedge \neg \square_bp)$ is true if and only if there is a method to inform $a$ of the truth of $p$ without informing $b$, under the assumption that the method itself is common knowledge. Or, in other (and slightly trendier) words, $[\aaul](\square_ap \wedge \neg \square_bp)$ is true if and only if it is possible to inform $a$ but not $b$ of the truth of $p$, without relying on \emph{security through obscurity}.

For more examples of the applications of arrow updates and arbitrary arrow updates, see \cite{AUL} and \cite{AAULAIJ}.

\begin{remark}
In the semantics of AAULC, we let $[\aaul]$ quantify over arrow updates in $\mathcal{L}_\mathit{AULC}$, so these updates may contain the common knowledge operator $C$. This means that our $[\aaul]$ operator is slightly different from the one in AAUL \cite{AAULAIJ}, since the quantification in AAUL is over updates that do not contain $C$.

This difference is not important for the current paper. All the results presented here still hold if we let $[\aaul]$ quantify only over the updates that contain neither $[\aaul]$ nor $C$.
\end{remark}

\section{Turing Machines}
\label{sec:machines}
A full discussion of Turing machines is outside the scope of this paper. We assume that the reader is familiar with the basic ideas of a Turing machine; here we only concern us with the notation that we use to represent Turing machines.

\begin{definition}
A \emph{Turing machine} $T$ is a tuple $T=(\Lambda,S,\Delta)$, where $\Lambda$ is a finite alphabet such that $\alpha_0\in \Lambda$, $S$ is a finite set of states such that $s_0,s_\mathit{end}\in S$ and $\Delta:\Lambda\times S\rightarrow \Lambda\times S\times \{\mathit{left},\mathit{remain},\mathit{right}\}$ is a transition function.
\end{definition}
We write $\Delta_1,\Delta_2$ and $\Delta_3$ for the projections of $\Delta$ to its first, second and third components. So if $\alpha$ is the symbol currently under the read/write head and $s$ is the current state, then the machine will write the symbol $\Delta_1(\alpha,s)$, go to the state $\Delta_2(\alpha,s)$ and move the read/write head in direction $\Delta_3(\alpha,s)$.

We assume, without loss of generality, that the state $s_0$ doesn't re-occur. Furthermore, note that we defined $\Delta$ to be a function with $\Lambda\times S$ as domain. So the machine $T$ continues after reaching $s_\mathit{end}$. This is notationally more convenient than letting $T$ terminate once it reaches $s_\mathit{end}$. We don't care about what happens after reaching $s_\mathit{end}$, though.

\begin{definition}
A Turing machine $T$ \emph{halts} if, when starting in state $s_0$ with a tape that contains only the symbol $\alpha_0$, the system reaches the state $s_\mathit{end}$.
\end{definition}
It is well known that the halting Turing machines are recursively enumerable, but the non-halting ones are not \cite{turing1937}. 

The Turing machines that we consider are deterministic, so the execution of a machine $T$ on a tape that only contains $\alpha_0$ happens in exactly one way. We call this the run of $T$. One straightforward way to represent this run of $T$ is to consider it as a function $\mathit{run}^T : \mathbb{Z}\times \mathbb{N}\rightarrow \Lambda\times S \times \{0,1\}$, where $\mathit{run}^T(n,m)=(\alpha,s,x)$ means that at time $m$, the symbol in position $n$ on the tape is $\alpha$, the machine is in state $s$ and the read/write head is at position $n$ if and only if $x=1$.

For notational reasons, it is convenient to extend this function to $\mathit{run}^T : \mathbb{Z}\times \mathbb{Z}\rightarrow \Lambda\times S \times \{0,1\}$, where $\mathit{run}^T(n,m)=(\alpha_0,s_\mathit{void},0)$ for all $m<0$. Doing so allows us to avoid a number of special cases that we would otherwise have to consider for $m=0$. Like with $\Delta$, we use $\mathit{run}^T_1$, $\mathit{run}^T_2$ and $\mathit{run}^T_3$ to refer to the projections to the first, second and third coordinates.

\section{The Reduction}
\label{sec:reduction}
For every Turing machine $T$, we want to represent the unique run $\mathit{run}^T$ in AAULC. In order to do this, we start by encoding certain facts as propositional atoms. For every state $s\in S\cup \{s_\mathit{void}\}$ and every element $\alpha\in \Lambda$ of the alphabet, we assume that $s,\alpha\in \mathcal{P}$. We are free to do this, since $\mathcal{P}$ is countably infinite, while $S$ and $\Lambda$ are finite. As one might expect, we use the propositional atom $s$ to represent the state of the Turing machine at a particular point in time being $s$, and we use the atom $\alpha$ to represent a particular position of the tape containing the symbol $\alpha$ at a particular point in time. Additionally, we assume that $\mathit{pos},\mathit{lpos},\mathit{rpos}\in \mathcal{P}$. These three atoms are used to indicate that a particular point on the tape is the current position of the read/write head, to the left of the current position of the read/write head and to the right of the current position of the read/write head, respectively.

We also assume that there are five agents named $a,\mathit{right},\mathit{left},\mathit{up}$ and $\mathit{down}$ in $\mathcal{A}$. Note that we can do this because we assumed that $|\mathcal{A}|\geq 5$. With these preliminaries out of the way, we can define the formula $\phi_T$ that represents the Turing machine $T$ in AAULC.

\begin{definition}
Let $T=(\Lambda,S,\Delta)$ be a Turing machine. The formula $\phi_T$ is given by
\begin{equation*}\phi_T := C\psi_\mathit{grid}\wedge C\psi_\mathit{sane}\wedge C\psi_T\wedge s_0\wedge \mathit{pos},\end{equation*}
where $\psi\mathit{grid}$, $\psi_\mathit{sane}$ and $\psi_T$ are as shown in Tables~\ref{table:psi_grid}--\ref{table:psi_T}.
\end{definition}

\begin{table*}[h]
\caption{The Formula $\psi_\mathit{grid}$.}
\label{table:psi_grid}
$$\begin{array}{rl}
\psi_\mathit{grid} := {} & \refa \wedge\mathit{direction}\wedge \mathit{inverse}\wedge\mathit{commute} \\
D := {} & \{\mathit{left},\mathit{right},\mathit{up},\mathit{down}\}\\
\mathit{INV} := {} & \{(\mathit{left},\mathit{right}),(\mathit{right},\mathit{left}),(\mathit{up},\mathit{down}),(\mathit{down},\mathit{up})\}\\
\mathit{COMM} := {} & (\{\mathit{up},\mathit{down}\}\times \{\mathit{left},\mathit{right}\})\cup (\{\mathit{left},\mathit{right}\}\times \{\mathit{up},\mathit{down}\})\\
\refa := {} & \lozenge_a\top\wedge [\aaul]\square_a\lozenge_a\top\\
\mathit{no\_other} := {} & \bigwedge_{x\in \agents\setminus (D \cup \{a\})}\square_x\bot\\
\mathit{direction} := {} &\bigwedge_{x\in D}(\lozenge_x\top \wedge [\aaul](\lozenge_x\lozenge_a\top\rightarrow \square_x\lozenge_a\top))\\
\mathit{inverse} := {} & [\aaul](\lozenge_a\top\rightarrow \bigwedge_{(x,y)\in \mathit{INV}}\square_x\square_y\lozenge_a\top)\\
\mathit{commute} := {} & [\aaul]\bigwedge_{(x,y)\in \mathit{COMM}}(\lozenge_x\lozenge_y\lozenge_a\top\rightarrow \square_y\square_x\lozenge_a\top)\\
\end{array}$$
\end{table*}
\begin{table*}[h]
\caption{The Formula $\psi_\mathit{sane}$.}
\label{table:psi_sane}
$$\begin{array}{rl}
\psi_\mathit{sane} := & \mathit{position}_1\wedge \mathit{position}_2\wedge\mathit{one\_state}\wedge \mathit{same\_state} \wedge \\
 &\mathit{one\_symbol}\wedge \mathit{void\_state}\wedge \mathit{initial\_symbol} \wedge \mathit{unchanged}\\
\mathit{position}_1 := & \neg (\mathit{pos}\wedge \mathit{lpos}) \wedge \neg (\mathit{pos}\wedge \mathit{rpos})\wedge \neg (\mathit{rpos}\wedge \mathit{lpos})\\
\mathit{position}_2 := & ((\mathit{pos}\vee \mathit{rpos})\rightarrow \square_\mathit{right}\mathit{rpos}) \wedge ((\mathit{pos}\vee \mathit{lpos})\rightarrow \square_\mathit{left}\mathit{lpos})\\
\mathit{one\_state} := & \bigvee_{s\in \mathit{states}}(s\wedge \bigwedge_{s'\in \mathit{states}\setminus \{s\}}\neg s)\\
\mathit{same\_state} := & \bigwedge_{s\in \mathit{states}}(s\rightarrow (\square_\mathit{left}s\wedge \square_\mathit{right}s))\\
\mathit{one\_symbol} := & \bigvee_{\alpha\in\mathit{symbols}}(\alpha \wedge \bigwedge_{\beta\in \mathit{symbols}\setminus \{\alpha\}}\neg \beta)\\
\mathit{void\_state} := & (s_0\vee s_\mathit{void})\rightarrow \square_\mathit{down}s_\mathit{void}\\
\mathit{initial\_symbol} := {} & s_0\rightarrow \alpha_0 \\
\mathit{unchanged} := {} & \bigwedge_{\alpha \in \mathit{symbols}}((\neg \mathit{pos}\wedge \alpha)\rightarrow\square_\mathit{up}\alpha) \\
\end{array}$$
\end{table*}
\begin{table*}[h]
\caption{The Formula $\psi_T$.}
\label{table:psi_T}
$$\begin{array}{rl}
\psi_T := {} & \mathit{position\_change}_T\wedge \mathit{state\_change}_T\wedge \mathit{symbol\_change}_T\\\
\mathit{position\_change}_T := {} & \bigwedge_{\{(s,\alpha)\mid \Delta_3(s,\alpha)=\mathit{left}\}}((\mathit{pos}\wedge s\wedge \alpha) \rightarrow \square_\mathit{up}\square_\mathit{left}\mathit{pos})\wedge \\ 
& \bigwedge_{\{(s,\alpha)\mid \Delta_3(s,\alpha)=\mathit{right}\}}((\mathit{pos}\wedge s\wedge \alpha) \rightarrow \square_\mathit{up}\square_\mathit{right}\mathit{pos})\wedge\\
& \bigwedge_{\{(s,\alpha)\mid \Delta_3(s,\alpha)=\mathit{remain}\}}((\mathit{pos}\wedge s\wedge \alpha) \rightarrow \square_\mathit{up}\mathit{pos})\\
\mathit{state\_change}_T :={} & \bigwedge_{s'\in\mathit{states}}\bigwedge_{\{(s,\alpha)\mid \Delta_2(s,\alpha)=s'\}}((\mathit{pos}\wedge s\wedge \alpha)\rightarrow \square_\mathit{up}s')\\
\mathit{symbol\_change}_T := {} & \bigwedge_{\beta\in \mathit{symbols}}\bigwedge_{\{(s,\alpha)\mid \Delta_1(s,\alpha)=\beta\}}((\mathit{pos}\wedge s\wedge \alpha)\rightarrow \square_\mathit{up}\beta)\\
\end{array}$$
\end{table*}

This formula may look somewhat intimidating, but apart from $\psi_\mathit{grid}$ all named formulas are very simple encodings of aspects of a Turing machine. The formula $\psi_\mathit{grid}$, as the name might suggest, encodes a $\mathbb{Z}\times\mathbb{Z}$ grid.

We first show that $\phi_T$ is satisfiable. After that, we show that any model that satisfies $\phi_T$ contains a representation of $\mathit{run}_T$.

\begin{lemma}
\label{lemma:sat}
For every Turing machine $T$, the formula $\phi_T$ is satisfiable.
\end{lemma}
\begin{proof}
Let $\M=(W,R,V)$ be given as follows. Take $W=\mathbb{Z}\times \mathbb{Z}$. For every direction $x\in D$ let $(w,w')\in R(x)$ if and only if $w'$ is immediately to the $x$ of $w$, and let $R(a)=\{(w,w)\mid w\in W\}$. For $x\not \in D\cup \{a\}$, let $R(x)=\emptyset$. Now, for any $\alpha\in \Lambda$ and $s\in S\cup\{s_\mathit{void}\}$, let $V(\alpha)=\{(n,m)\mid \mathit{run}^T_1(n,m)=\alpha\}$ and $V(s)=\{(n,m)\mid \mathit{run}^T_2(n,m)$. Furthermore, let $V(\mathit{pos}) =\{(n,m)\mid \mathit{run}^T_3(n,m)=1\}$, $V(\mathit{lpos})=\{(n,m)\mid \exists n'>n : \mathit{run}^T_3(n',m)=1\}$ and $V(\mathit{rpos})=\{(n,m)\mid \exists n'<n : \mathit{run}^T_3(n',m)=1\}$. (In other words, $\mathit{lpos}$ holds if you are to the left of $\mathit{pos}$ and $\mathit{rpos}$ holds if you are to the right of $\mathit{pos}$.)

We claim that $\M,(0,0)\models \phi_T$. Since $\mathit{run}^T(0,0)=(\alpha_0,s_0,1)$, we have $\M,(0,0)\models \mathit{pos}\wedge\mathit{s_0}$. This leaves the conjuncts $C\psi_\mathit{grid}, C\psi_\mathit{sane}$ and $C\psi_T$. We start by looking at $C\psi_\mathit{sane}$.

The conjuncts of $\psi_\mathit{sane}$ hold under the following conditions.
\begin{itemize}
	\item $\mathit{position}_1$ holds if being the position of the head, being to the right of the head and being to the left of the head are mutually exclusive.
	\item $\mathit{position}_2$ holds if all worlds to the left of the head satisfy $\mathit{lhead}$ and all worlds to the right of the head satisfy $\mathit{rhead}$.
	\item $\mathit{initial\_symbol}$ holds if, at the initial state $s_0$, the entire tape contains the symbol $\alpha_0$.
	\item $\mathit{one\_state}$ holds if at every $(n,m)$, the system is in exactly one state.
	\item $\mathit{same\_state}$ holds if for every $n,m,k\in \mathbb{Z}$, the worlds $(n,m)$ and $(k,m)$ are in the same state. (So the state depends only on time, not on the tape position.)
	\item $\mathit{one\_symbol}$ holds if at every time $m$, every position $n$ contains exactly one symbol.
	\item $\mathit{void\_state}$ holds if at every time before $s_0$, the system was in the dummy state $s_\mathit{void}$.
	\item $\mathit{unchanged}$ holds if every symbol that is not under the read/write head remains unchanged.
\end{itemize}
All of these conditions are satisfied, because the valuation of $\M$ was derived from the run of a Turing machine. So $\M,(0,0)\models C\psi_\mathit{sane}$. 
Now, consider the conjuncts of $\psi_T$.
\begin{itemize}
	\item $\mathit{position\_change}_T$ holds if the read/write head moves in the appropriate direction, as specified by $\Delta$.
	\item $\mathit{state\_change}_T$ holds if the state changes as specified by $\Delta$.
	\item $\mathit{symbol\_change}_T$ holds if the symbol under the read/write head is written as specified by $\Delta$.
\end{itemize}
These conditions are also satisfied, because the valuation of $\M$ was derived from $\mathit{run}_T$. So $\M,(0,0)\models C\psi_T$. Left to show is that $\M,(0,0)\models \psi_\mathit{grid}$.

So take any $w\in W$. The world $w$ itself is the only $a$-successor of $w$. So we have $\M,w\models \lozenge_a\top$. Furthermore, it is impossible for any arrow update to retain the $a$-arrow from $w$ while removing the $a$-arrows from its successor, since they are the same $a$-arrow. It follows that $\M,w\models [\aaul]\square_a\lozenge_a\top$. So we have shown that $\M,w\models \refa$.

We have defined $R_x=\emptyset$ for all $x\not \in D\cup \{a\}$, so we also have $\M,w\models \mathit{no\_other}$.

Now, consider $\mathit{direction}$. Take any $x\in D$. There is a $x$-arrow from $w$ to the world $w'$ to its $x$, so $\M,w\models \lozenge_x\top$. Furthermore, since this $w'$ is the only $x$-successor of $w$, it follows that it is impossible to retain an $a$-arrow on one $x$-successor of $w$ while removing all $a$-arrows from another. So $\M,w\models [\aaul](\lozenge_x\lozenge_a\top\rightarrow\square_x\lozenge_a\top)$. This holds or every $x\in D$, so $\M,w\models\mathit{direction}$.

For opposite directions $x$ and $y$, there is exactly one $x$-$y$-successor of $w$, namely $w$ itself. It follows that it is impossible for any arrow update to retain the $a$-arrow on $w$ while removing all $a$-arrows from its $x$-$y$-successor, so $\M,w\models\mathit{inverse}$.

Finally, for perpendicular directions $x$ and $y$ there is exactly one $x$-$y$-successor $w'$ of $w$, and this $w'$ is also the unique $y$-$x$-successor of $w$. So it is impossible for an arrow update to retain the $a$-arrow from the $x$-$y$-successor while removing it from some $y$-$x$-successor. So $\M,w\models \mathit{commute}$.

This completes the proof that $\M,w\models \psi_\mathit{grid}$ and therefore the proof that $\M,w\models \phi_T$. So $\phi_T$ is satisfiable.
\end{proof}

\begin{lemma}
\label{lemma:reduction}
If $\M,w_0\models \phi_T$, then $\M,w_0\models C\neg s_\mathit{end}$ if and only if $T$ is non-halting.
\end{lemma}
\begin{proof}
Suppose $\M,w_0\models \phi_T$. Then, by definition, $\M,w_0\models\psi_\mathit{grid}\wedge C\psi_\mathit{sane}\wedge C\psi_T\wedge \mathit{pos}\wedge s_0$. We will first show that $\M,w_0\models C\psi_\mathit{grid}$ implies that the model $\M$ is grid-like. Then, we will show that the remaining subformulas imply that the model $\M$ represents $\mathit{run}^T$.

By the $\lozenge_a\top$ conjunct of $\refa$, every reachable world $w$ has at least one $a$-successor. If any $a$-successor of $w$ is $\mathit{AULC}$-distinguishable from $w$, then it would be possible for an arrow update to remove the $a$-arrow from this successor while retaining the $a$-arrow from $w$. This would contradict the $[\aaul]\square_a\lozenge_a\top$ conjunct of $\refa$.

Now, for any $x\in D$, consider the $x$-successors of $w$, of which there is at least one by the $\lozenge_x\top$ part of $\mathit{direction}$. If any of these successors were $\mathit{AULC}$-distinguishable, it would be possible to remove the $a$-arrow from one of them but not from the other. This would contradict the $[\aaul](\lozenge_x\lozenge_a\top\rightarrow \square_x\lozenge_a\top)$ part of $\mathit{direction}$.

For any opposite directions $x$ and $y$, consider any $x$-$y$-successor $w'$ of $w$. If any $\mathit{AULC}$ formula could distinguish between $w'$ and $w'$,
it would be possible for an arrow update to retain the $a$-arrow from $w$ while removing the $a$-arrow from $w'$, which would contradict $\mathit{inverse}$.

Finally, for any perpendicular direction $x$ and $y$, consider any $x$-$y$- and $y$-$x$-successors of $w$. If these successors were $\mathit{AULC}$-distinguishable, it would be possible for an arrow update to remove one while retaining the other, contradicting $\mathit{commute}$.

Taken together, the above facts imply that $\M$ contains a representation of a grid $\mathbb{Z}\times\mathbb{Z}$ where, for every $x\in D$, we have $w'\in R_x(w)$ if and only if $w'$ is to the $x$ of $w$. Furthermore, if $w$ represents $(n,m)$ then so does every $a$-successor of $w$. Every world $(n,m)$ may be represented by multiple worlds in $\M$, but all the worlds that represent a single grid point are $\mathit{AULC}$-indistinguishable from one another. Furthermore, the formula $\mathit{no\_other}$ implies that we cannot escape this grid, every reachable world represents some grid point $(n,m)$.

The remaining subformulas of $\phi_T$ guarantee that this grid encodes $\mathit{run}_T$. First, consider $\psi_\mathit{sane}$. This formula enforces a number of general sanity constraints.
\begin{itemize}
	\item The formula $\mathit{position}_1$ says that being the current position of the read/write head, being to the right of the position of the head and being to the left of the position of the head are mutually exclusive. 
	\item The formula $\mathit{position}_2$ says that if you are either at the current position of the read/write head or to the right of the current position, then if you go further to the right then you will be to the right of the current position. Similarly, it says that if you are either at the current position of the head or to the left of it and go further left, then you will be to the left of the head. Together with $\mathit{position}_1$, this guarantees that the read/write head is in at most one position at any time step. (Ensuring that the head is in at least one position at every time is done later).
	\item The formula $\mathit{one\_state}$ says that every world is in exactly one state.
	\item The formula $\mathit{same\_state}$ says that if a world is in state $s$, then the worlds to the left and right are also in state $s$. So all the worlds that represent a single time step satisfy the same state. Together with $\mathit{one\_state}$, this implies that every time step is associated with exactly one state.
	\item The formula $\mathit{one\_symbol}$ says that every world satisfies exactly one symbol.
	\item The formula $\mathit{void\_state}$ says that every time before the initial state $s_0$ is in the dummy state $s_\mathit{void}$. So the worlds satisfying $s_0$ are where the computation starts.
	\item The formula $\mathit{initial\_symbol}$ says that the $s_0$ worlds satisfy $\alpha_0$, so if the system is in the initial state $s_0$, then the tape is empty.
	\item Finally, the formula $\mathit{symbol\_unchanged}_T$ guarantees that the symbol remains unchanged everywhere other than under the read/write head.
\end{itemize}
Now, consider $\psi_T$, which forces the transitions to satisfy $\Delta$.
\begin{itemize}
	\item The formula $\mathit{position\_change}_T$ guarantees that the read/write head moves in the correct direction, depending on the current symbol under the head and the current state.
	\item The formula $\mathit{state\_change}_T$ guarantees that the next state is as specified by $T$.
	\item The formula $\mathit{symbol\_change}_T$ guarantees that the correct symbol is written to the tape, as specified by $T$.
\end{itemize}
The last two conjuncts of $\phi_T$ do not contain a common knowledge operator. They state that the world $w_0$ satisfies $\mathit{pos}$ and $s_0$. So $w_0$ represents the point $(0,0)$. Because the rules of $T$ always require the read/write head to stay in he same position or to move to the left or right, this also implies that at every time after $w_0$ the head is in at least one position.

Taken together, the above shows that the valuation on the grid represents $\mathit{run}_T$. So the grid contains a $s_\mathit{end}$ state if and only if $T$ is halting. Since every reachable state is part of the grid, it follows that $\M,w_0\models C\neg s_\mathit{end}$ if and only if $T$ is non-halting.
\end{proof}

\begin{theorem}
The formula $\phi_T\rightarrow C\neg s_\mathit{end}$ is valid if and only if $T$ is non-halting. Furthermore, $\phi_T\rightarrow \neg C\neg s_\mathit{end}$ is valid if and ony if $T$ is halting.
\end{theorem}
\begin{proof}
Suppose that $T$ is non-halting. Then, by Lemma~\ref{lemma:reduction}, we have $\models \phi_T\rightarrow C\neg s_\mathit{end}$. Furthermore, since $\phi_T$ is satisfiable, this implies that $\not\models \phi_T\rightarrow \neg C\neg s_\mathit{end}$.

Suppose, on the other hand, that $T$ is halting. Then, by Lemma~\ref{lemma:reduction}, we have $\models \phi_T\rightarrow \neg C\neg s_\mathit{end}$. Furthermore, since $\phi_T$ is satisfiable, this implies that $\not \models \phi_T\rightarrow C\neg s_\mathit{end}$.
\end{proof}

\begin{corollary}
The set of valid formulas of AAULC is neither RE not co-RE.
\end{corollary}

\begin{corollary}
AAULC does not have a finitary axiomatization.
\end{corollary}

\section{Conclusion}
The validity problems for the quantified update logics APAL, GAL, CAL and AAUL are known not to be co-RE. It is not currently known whether these problems are RE. This question is particularly relevant because if the validity problem of a logic is not RE, then that logics cannot have a recursive axiomatization.

The logic AAULC adds a common knowledge operator to AAUL. Here, we showed that the validity problem of AAULC is not RE, using a reduction from the non-halting problem of Turing machines. This reduction uses the common knowledge operator $C$, so it does not immediately follow that the validity problem of AAUL is not RE. Still, we believe that the proof presented here can be adapted for AAUL.

It is less clear whether our reduction could be adapted for APAL, GAL and CAL. Still, it seems worthwhile to attempt to modify this reduction for APAL, GAL and CAL. If such an attempt succeeds, it would show that these logics are nor recursively axiomatizable. Or if the attemp fails, then the way in which it fails might provide a hint about how to prove that the validity problems of these logics are RE.

\bibliographystyle{eptcs}
\bibliography{arbitraryarrows}

\end{document}